%% file: root.tex
\documentclass[letterpaper, 10 pt, conference]{ieeeconf}  

\usepackage{graphicx}
\usepackage{setspace}

\IEEEoverridecommandlockouts                              
\title{\LARGE \bf
Agile Temporal Discretization for Symbolic Optimal Control
}

\author{Adrien Janssens \and Adrien Banse \and Julien Calbert \and Raphaël M. Jungers
\thanks{R. M. Jungers is a FNRS honorary Research Associate. This project has received funding from the European Research Council (ERC) under the European Union's Horizon 2020 research and innovation programme under grant agreement No 864017 - L2C, from the Horizon Europe programme under grant agreement No101177842 - Unimaas, and from the ARC (French Community of Belgium)- project name: SIDDARTA. Adrien Banse is supported by the French Community of Belgium in the framework of a FNRS/FRIA grant.}
\thanks{A. Janssens, A. Banse, J. Calbert and R. Jungers are with the ICTEAM, UCLouvain (Louvain-la-Neuve, Belgium). E-mail adresses: \texttt{adrien.janssens@student.uclouvain.be}, \texttt{\{adrien.banse, julien.calbert, raphael.jungers\}@uclouvain.be}%
}}

\usepackage{soul}
\usepackage{booktabs}

\input{header}
\input{macros}

\begin{document}
\maketitle
\thispagestyle{empty}
\pagestyle{empty}

\newtheorem{remark}{Remark}

\begin{abstract}
As control systems grow in complexity, abstraction-based methods have become essential for designing controllers with formal guarantees. However, a key limitation of these methods is their reliance on discrete-time models, typically obtained by discretizing continuous-time systems with a fixed timestep. This discretization leads to two major problems: when the timestep is small, the abstraction includes numerous stuttering and spurious trajectories, making controller synthesis suboptimal or even infeasible; conversely, a large time step may also render control design infeasible due to a lack of flexibility. In this work, drawing inspiration from Reinforcement Learning concepts, we introduce temporal abstractions, which allow for a flexible timestep. We provide a method for constructing such abstractions and formally establish their correctness in controller design. Furthermore we show how to apply these to optimal control under reachability specifications. Finally we showcase our methods on two numerical examples, highlighting that our approach leads to controllers that achieve a lower worst-case control cost.
\end{abstract}


\section{INTRODUCTION} \label{sec:intro}

The growing complexity of engineered dynamical systems has made controller design increasingly arduous. In the literature, this paradigm shift has been coined by the academic and industrial communities as the \emph{cyber-physical revolution} \cite{kim2012cyber, alur2015principles, lee2009introducing}. On the other hand, it is crucial to provide safety and/or performance guarantees for such systems \cite{knight2002safety}. In particular, we focus on \emph{reachability analysis} \cite{althoff2021set}: designing a controller that drives the system from an initial set to a target set while avoiding obstacles. Being able to solve such problems for complex dynamical systems is crucial in many applications (see \cite{althoff2014online,gillula2011applications} for applications to autonomous driving and robotic aerial vehicles). 

In this context, \emph{abstraction-based control}, also known as symbolic control, is a promising tool to provide formal reachability guarantees \cite{tabuada2009verification}. Every abstraction-based control method follows the same systematic procedure: 
\begin{enumerate}
    \item \emph{Abstraction} -- The concrete system is approximated by a finite-state mathematical object (the \emph{abstract} system); 
    \item \emph{Abstract controller synthesis} -- A discrete controller is designed for the abstract system;  
    \item \emph{Concretization} -- From the discrete controller, a concrete controller is refined for the original (\emph{concrete}) dynamical system.
\end{enumerate}
The validity of this approach and the procedure to concretize the discrete controller are provided by a mathematical \emph{simulation relation}~\cite{tabuada2009verification,calbert2024classification} between the abstract and concrete systems, such as the \emph{feedback refinement relation} ($\FRR$)~\cite{reissig2016feedback}.

Continuous-time deterministic dynamical systems are abstracted by partitioning the state space into a finite set of \emph{cells}, selecting a finite set of inputs, considering discrete-time dynamics with timestep $\tau$. The abstract states correspond to the cells, and a transition between $s_1$ and $s_2$ with input $u$ is added if the set of reachable states from the cell $s_1$ intersects with the cell $s_2$. In this work, we focus on the time discretization. Indeed, in the context of reachability analysis, fixing a too small constant timestep $\tau$ leads to conservatism for at least two reasons: 
\begin{enumerate}
    \item \emph{Stuttering transitions} -- If the timestep is small, then the reachable set from an initial cell may intersect with the cell itself, leading to a self-loop, or \emph{stuttering transition} in the automaton. When performing reachability analysis on the abstract system, a possible behaviour is that the state stays infinitely in the initial cell, whereas this may be impossible in reality. This kind of transition is illustrated in yellow on Figure~\ref{fig:motiv}.
    \item \emph{Spurious trajectories} -- At each step of a trajectory in the abstract system, the information of where \emph{exactly} the state is within the cell is lost. As a consequence, an error is propagated with the number of steps, and more and more trajectories are valid in the abstraction but cannot happen in reality: such trajectories are called \emph{spurious trajectories}. This phenomenon is also illustrated in red (two small timesteps) and in blue (one long timestep) in Figure~\ref{fig:motiv}.
\end{enumerate}
\begin{figure}
    \centering
    \includegraphics[width=0.85\linewidth]{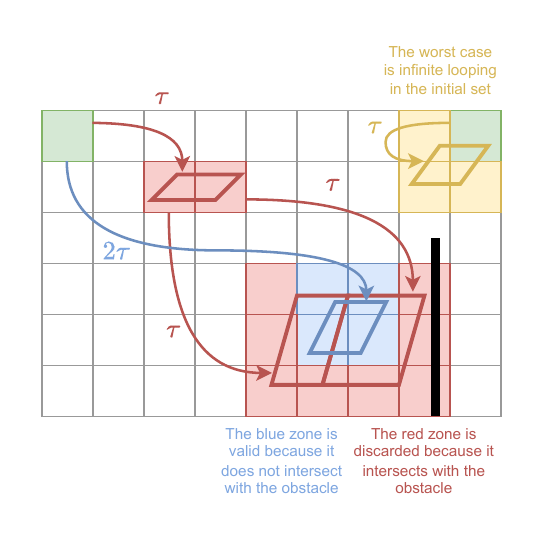}
    \caption{Illustration of stuttering transitions (yellow) and spurious trajectories (red and blue). \textbf{Yellow}: the timestep is too small, and therefore the reachable set intersects with the starting cell. \textbf{Red}: from the initial green cell, two cells are reachable with a jump of length $\tau$. At the next step, the two reachable sets from both cells are computed. Applying the same input for $2\tau$ is discarded because 3 cells intersect with the black obstacle. \textbf{Blue}: we now directly consider a timestep corresponding to a larger timestep, $2\tau$. The reachable set intersects with 4 cells that do not intersect with the obstacle, making the input valid.}
    \label{fig:motiv}
\end{figure}
This conservatism leads either to infeasible problems, or to sub-optimal solutions in the context of optimal control. On the other hand, considering small timesteps allows for more flexibility, and is sometimes needed in order to ensure feasibility of the problem. 

\paragraph*{Contributions} In this work, we propose to take into account this trade-off by introducing \emph{temporal abstractions}, in which multiple timesteps are allowed. We prove the correctness of such abstractions, thereby providing formal guarantees. We then provide the necessary theoretical framework to use these abstractions in the context of optimal control under reachability specifications. While our approach does not directly tackle the curse of dimensionality inherent to abstraction-based control, numerical experiments show that using multiple timesteps yields lower worst-case cost controllers than those obtained with a finer input space discretization.

\paragraph*{Related works} Our work takes inspiration from the notion of \emph{temporal abstraction} in the Reinforcement Learning community (see e.g. \cite{sutton1999between} for a survey), but aims at combining this powerful framework and the formal guarantees that abstraction-based methods can provide. The works \cite{ivanova2020lazy, kader2019safety} have already investigated flexible timesteps in the context of abstraction-based techniques but for safety analysis, and with a different abstraction construction.

\paragraph*{Outline} The rest of this paper is organized as follows. In Section~\ref{sec:preliminaries}, we introduce all the necessary formalism to ensure correctness of the approach. In Section~\ref{sec:main}, we introduce the concept of \emph{temporal abstractions} and how to apply them to an optimal control problem. Finally, in Section~\ref{sec:experiments} we demonstrate the quality of our method on two examples.

\paragraph*{Notations}
%
Given two sets $\set{X}_1,\set{X}_2$, we define a \emph{single-valued map} as $f:\set{X}_1\rightarrow \set{X}_2$, while a~\emph{set-valued map} is defined as $F:\set{X}_1\rightarrow 2^{\set{X}_2}$, where $2^{\set{X}_2}$ is the power set of $\set{X}_2$, i.e., the set of all subsets of $\set{X}_2$. 
The image of a subset $\Omega\subseteq \set{X}_1$ under $F:\set{X}_1\rightarrow 2^{\set{X}_2}$ is denoted $F(\Omega)$. 
Given two sets \(\set{X}_1\) and \(\set{X}_2\), we identify a binary \emph{relation} \(R \subseteq \set{X}_1 \times \set{X}_2\) with set-valued maps, i.e., \(R(x_1) = \{x_2 \mid (x_1, x_2) \in R\}\) and \(R^{-1}(x_2) = \{x_1 \mid (x_1, x_2) \in R\}\). A relation \(R \subseteq \set{X}_1 \times \set{X}_2\) is said to be \emph{strict} if \(R(x_1) \neq \myemptyset\) for all \(x_1 \in \set{X}_1\). 
When \(\set{X}_2\) is finite, the mapping \(R: \set{X}_1 \to 2^{\set{X}_2}\) is referred to as a \emph{quantizer}.
%
The set of all signals that take their values in $\set{B}$ are defined on intervals of the form $[0;N) =[0;N-1]$ is denoted $\set{B}^{\infty}$, $\set{B}^{\infty} = \bigcup_{N\in \N\cup \{\infty\}}\set{B}^{[0;N)}$. 
Given a set-valued map $f:\set{X}\rightarrow 2^{\set{Y}}$ and $\vect{x}\in \set{X}^{[0;N)}$, we denote by 
$f(\vect{x}) = \{\vect{y}\in \set{Y}^{[0;N)} \mid \forall k\in [0;N):\ y(k)\in f(x(k))\}$.

\section{PRELIMINARIES}\label{sec:preliminaries}
In this section, we introduce the control formalism and the symbolic control framework required to ensure the correctness of the proposed approach.
\subsection{Control framework}
\begin{definition} \label{def:transition_time_system}
    A \textit{transition system} is a triple $\Sys=(\set{X}, \set{U}, F)$, where $\set{X}$ and $\set{U}$ are respectively the sets of \emph{states} and \emph{inputs}, and $F: \set{X} \times \set{U} \to 2^{\set{X}}$ is the \emph{transition map} such that
    \begin{equation}\label{eq:transition_map}
        x(k+1) \in F(x(k), u(k)).
    \end{equation} 
    A tuple $(\vect x,\vect u)\in \set{X}^{[0;N)}\times \set{U}^{[0;N)}$ is a \emph{trajectory} of length $N$ of $\Sys$ starting at $x(0)$ if $N\in\N\cup \{\infty\}$ and~\eqref{eq:transition_map} holds for all $k\in [0;N-1)$.~\hfill~$\triangle$
\end{definition}

Given a transition system $\Sys = (\set{X}, \set{U}, F)$, we introduce the set-valued operator of \emph{available inputs}, defined as \(\set{U}_{\Sys}(x) = \{ u \in \set{U} \mid F(x, u) \ne \myemptyset \}\), which gives the set of inputs \(u\) available at a given state \(x\).
For a given \(u \in \set{U}\), the \emph{reachable set} of \(\Omega \subseteq \set{X}\) is defined as  
$F(\Omega, u) := \bigcup_{x \in \Omega} F(x, u)$.

We use a set-valued function to model perturbations and non-determinism in a unified formalism. We say that the transition system is deterministic if for every state \(x \in \set{X}\) and input \(u \in \set{U}\), \(F(x, u)\) is either empty or a singleton

%
A \emph{finite} transition system refers to a system characterized by finitely many states and inputs. We define the \emph{behaviour} of a system as its set of \emph{maximal} trajectories.
\begin{definition}
    Given the system $\Sys$ in Definition~\ref{def:transition_time_system}, the set \begin{multline}
        \set{B}(\Sys)=\{ (\vect x,\vect u) \mid (\vect x,\vect u)\text{ is a trajectory of }\Sys
        \text{ on } [0;N),\\
        \text{ and if }N<\infty \text{ then }
        F(x(N-1),u(N-1)) = \myemptyset\},
    \end{multline}
    is called the \emph{behaviour} of $\Sys$.
    \hfill $\triangle$
\end{definition}
For clarity, we focus on \emph{static} controllers, though most results extend naturally to dynamical controllers.
\begin{definition}
    A \emph{static controller} for a transition system \(\Sys = (\set{X}, \set{U}, F)\) is a set-valued map \(\Cont : \set{X} \to 2^{\set{U}}\) satisfying \(\forall x \in \set{X}, \ \Cont(x) \subseteq \set{U}(x)\). The \emph{domain} of \(\Cont\) is defined as \(\dom(\Cont) = \{x \in \set{X} \mid \Cont(x) \ne \myemptyset\}\). The \emph{controlled system} is the transition system \((\set{X}, \set{U}, F_{\Cont})\) where  
    \[
    x^+ \in F_{\Cont}(x, u) \Leftrightarrow u \in \Cont(x) \land x^+ \in F(x, u), 
    \]
    and is denoted \(\Cont \times \Sys\).
    \hfill $\triangle$
\end{definition}
This controller is \emph{static} because the allowed inputs at each state depend only on the current state, not on the past trajectory. We now define the control problem.
\begin{definition}\label{def:specification}
    Given a transition system \(\Sys = (\set{X},\set{U}, F)\), a \emph{specification} \(\Sigma\) is any subset \(\Sigma \subseteq (\set{X} \times \set{U})^{\infty}\). The system \(\Sys\) \emph{satisfies} \(\Sigma\) if \(\set{B}(\Sys) \subseteq \Sigma\). A \emph{control problem} is a pair \((\Sys, \Sigma)\), and a controller \(\Cont\) \emph{solves} it if \(\set{B}(\Cont \times \Sys) \subseteq \Sigma\).
    \hfill $\triangle$
\end{definition}

\subsection{Abstraction-based control}
Given the concrete system \(\Sys_1 = (\set{X}_1, \set{U}_1, F_1)\) and the abstract system \(\Sys_2 = (\set{X}_2, \set{U}_2, F_2)\), we aim to construct a concrete controller \(\Cont_1 : \set{X}_1 \rightarrow 2^{\set{U}_1}\) from an abstract controller \(\Cont_2 : \set{X}_2 \rightarrow 2^{\set{U}_2}\), as part of the third step of the abstraction approach described in the introduction. To do so, the relation \(R\subseteq \set{X}_1\times \set{X}_2\) must impose conditions on local dynamics, ensuring that input choices lead to consistent state transitions across associated states. One such relation is the \emph{feedback refinement relation}.

\begin{definition}[{\cite[Definition V.2]{reissig2016feedback}}]\label{def:FRR}
Given two transition systems $\Sys_1=(\set{X}_1, \set{U}_1 ,F_1)$ and $\Sys_2=(\set{X}_2, \set{U}_2, F_2)$, a relation $R \subseteq \set{X}_1\times \set{X}_2$ is a \emph{feedback refinement relation} from $\Sys_1$ to $\Sys_2$, denoted $\Sys_1\preceq_R^{\FRR} \Sys_2$,~if $\forall (x_1,x_2)\in R$, the conditions
\begin{gather}
\set{U}_{\Sys_2}(x_2)\subseteq \set{U}_{\Sys_1}(x_1) \\
\forall u\in\set{U}_{\Sys_2}(x_2):\  R(F_1(x_1, u))\subseteq  F_2(x_2,u)
\end{gather}
hold. \hfill $\triangle$
\end{definition}

The feedback refinement relation benefits from a simple concretization scheme as established by the following proposition.
\begin{proposition}[{\cite[Theorem~V.4]{reissig2016feedback}}]\label{prop:FRR}
    Let \(\Sys_1 =(\set{X}_1,\set{U}_1,F_1)\) and \(\Sys_2=(\set{X}_2,\set{U}_2,F_2)\) be transition systems, and let \(R \subseteq \set{X}_1 \times \set{X}_2\). If \(\Sys_1 \preceq_R^{\FRR} \Sys_2\), then for any controller \(\Cont_2\), the controller \(\Cont_1 = \Cont_2 \circ R\) ensures $\set{B}(\Cont_1 \times \Sys_1) \subseteq R^{-1}(\set{B}(\Cont_2 \times \Sys_2))$.
    Consequently, if \(\Sigma_1 \subseteq R^{-1}(\Sigma_2)\) and \(\Cont_2\) solves \((\Sys_2, \Sigma_2)\), then \(\Cont_1\) solves \((\Sys_1, \Sigma_1)\).
\end{proposition}

Altough $\Sys_2$ need not be finite, practical use of graph-theoretic tools for controller synthesis requires a finite abstraction (see Section~\ref{sec:experiments}).

\subsection{Continuous-time systems}
In this work, we focus on the control of continuous-time systems.
\begin{definition} \label{def:continuous_time_system}
    A \emph{continuous-time system} is a triple $\mathcal{D}=(\set{X}, \set{U}, f)$, where $\set{X} \subseteq \R^n$ and $\set{U}\subseteq \R^m$ are respectively the sets of \emph{states} and \emph{inputs}, and $f: \R^n \times \set{U} \to \R^n$ is the vector field such that
    \begin{equation}\label{eq:differential_equation}
        \dot{x}(t) = f(x(t), u(t)).
    \end{equation}
    We assume that $f(\cdot, u)$ is locally Lipschitz for all $u\in\set{U}$. 
    \hfill $\triangle$
\end{definition}

In the sequel, $\phi$ denotes the general solution of the system associated with~\eqref{eq:differential_equation} for
constant inputs. That is, if $x_0\in\R^n$, $u\in \set{U}$, and $f(\cdot, u)$ is locally Lipschitz, then $\phi(\cdot, x_0, u)$ is the unique solution of the initial value problem $\dot{x}=f(x,u)$, $x(0)=x_0$~\cite{hartman2002ordinary}.

\section{MAIN RESULTS} \label{sec:main}

In this section, we introduce \emph{temporal abstractions}. We start by providing formal definitions of both the concrete systems $\Sys_1$ and the abstract systems $\Sys_2$. We then provide conditions on the abstract system $\Sys_2$ for \(\Sys_1 \preceq_R^{\FRR} \Sys_2\) to hold (see Theorem~\ref{th:formal_abstraction}). We then show how to use the latter for optimal control of a system under reach-avoid specifications. We show that the developed abstraction-based techniques lead to optimality guarantees (see Theorem~\ref{th:controller_from_value_function}). 


\subsection{Temporal abstractions}
We begin by introducing a general framework to frame a continuous-time system as a transition system without loss of generality, by extending the input space:
\begin{definition} \label{def:multisam}
    Consider a continuous-time system \(\mathcal{D} = (\set{X}, \set{U}, f)\). This system can be framed as a transition system \(\Sys_1 = (\set{X}_1, \set{U}_1, F_1)\) as follows: 
    \begin{itemize}
        \item The state space is unchanged $\set{X}_1 = \set{X}$;
        \item The input space $\set{U}_1$ is defined as $\set{U}_1 := \set{U} \times [0;\infty)$; 
        \item The transition function $F_1 : \set{X}_1 \times \set{U}_1 \to 2^{\set{X}_1}$ is given by
        \begin{equation}\label{eq:continuous_to_transition}
            F_1(x, (u, t)) := \{\phi(t, x, u)\}.
        \end{equation}
    \end{itemize}
    We call this system the \emph{multi-sampled system}. \hfill $\triangle$
\end{definition}


Now, in this work, we apply abstraction-based control to the continuous-time system \(\mathcal{D} = (\set{X}, \set{U}, f)\) (see Definition~\ref{def:continuous_time_system}). Consequently, this approach involves constructing both a transition system \(\Sys_2 = (\set{X}_2, \set{U}_2, F_2)\) and a relation \(R \subseteq \set{X}_1 \times \set{X}_2\) such that \(\Sys_1 \preceq_R^{\FRR} \Sys_2\), where \(\Sys_1 = (\set{X}_1, \set{U}_1, F_1)\) is the transition system associated with \(\mathcal{D}\) as in Definition~\ref{def:multisam}. As proven below, the following conditions on the abstract system $\Sys_2$ and $R$ are sufficient for \(\Sys_1 \preceq_R^{\FRR} \Sys_2\) to hold:
\begin{itemize}
    \item The input space \(\set{U}_2\) is given by
    \begin{equation} \label{eq:abstract_input}
        \set{U}_2 = \set{U}_2' \times \set{T} \subseteq \set{U}_1 = \set{U} \times [0,\infty),
    \end{equation}
    where \(\set{U}_2'\subseteq \set{U}\) represents a subset of the concrete system's inputs, and \(\set{T} \subseteq [0,\infty)\) is a subset of times;  
    \item For all \(u_2 = (u, \tau) \in \set{U}_2\), the transition map satisfies \begin{equation} \label{eq:transition}
    \begin{aligned}
        &F_2(x_2, u_2) \\
        &\,\, \supseteq \{x_2^+ \mid F_1(R^{-1}(x_2), (u, \tau)) \cap R^{-1}(x_2^+) \neq \myemptyset \}.
    \end{aligned}
    \end{equation}
\end{itemize}

The finiteness of \(\set{U}_2\) implies that only a finite number of inputs and sampling times are selected from the original system. As a result, the abstract system \(\Sys_2\) can be interpreted as a \emph{sampled} version of~\(\Sys_1\). When \(\set{T}\) is a singleton, i.e., \(\set{T} = \{\tau\}\) for some \(\tau \in [0,\infty)\), it corresponds to the classical discretization setting with a fixed timestep \(\tau\).

\begin{remark} \label{rem:general_but_uniform}
In this section, we remain general. However, in Section~\ref{sec:experiments}, we will provide an explicit procedure for constructing such abstract system $\Sys_2$ and relation $R$. \hfill $\triangle$
\end{remark}

Now, we prove the correctness of our method by showing that, for any $\Sys_1, \Sys_2$ and $R$ such as defined above, \(\Sys_1 \preceq_R^{\FRR} \Sys_2\). For the remainder of this paper, we assume that \(R\) is strict (see \emph{Notations} paragraph in Section~\ref{sec:intro} for a definition).

\begin{theorem}\label{th:formal_abstraction}
    Let \(\mathcal{D} = (\set{X}, \set{U}, f)\) be a continuous-time system with its associated transition system \(\Sys_1 = (\set{X}_1, \set{U}_1, F_1)\) as in~\eqref{eq:continuous_to_transition}. Consider a transition system \(\Sys_2 = (\set{X}_2, \set{U}_2, F_2)\) and a strict relation \(R \subseteq \set{X}_1 \times \set{X}_2\), with $\set{U}_2$ as in \eqref{eq:abstract_input}, and $R$ and $F_2$ as in \eqref{eq:transition}. Then it holds that \(\Sys_1 \preceq_R^{\FRR} \Sys_2\).\hfill $\triangle$
\end{theorem}
\begin{proof}
    For all \((x_1, x_2) \in R\) and \((u, \tau) \in \set{U}_2(x_2)\), we have \(F_2(x_2, u_2) \neq \myemptyset\), which by construction ensures \((x, u) \in \set{U}_1(x_1)\). Given \(\{x_1^+\} = F(x_1, (u, \tau))\), the strictness of \(R\) implies \(R(x_1^+) \neq \myemptyset\) and \(R(x_1^+) \subseteq F_2(x_2, (u, \tau))\).
\end{proof}
\subsection{Application to optimal control}
We begin by defining the considered reach-avoid specification (see Definition~\ref{def:specification}). 
\begin{definition}\label{def:specific_spec}
Given a transition system $\Sys = (\set{X}, \set{U}, F)$, a \emph{reach-avoid} specification associated with sets $\set{I}, \set{F}, \set{O} \subseteq \set{X}$ such that $\set{F}\cap \set{O}= \myemptyset$ is defined as 
\begin{multline}\label{eq:reach_avoid}
\Sigma^{\text{Reach}} = \{ (\vect{x},\vect{u}) \in (\set{X}\times \set{U})^{\infty} \mid x(0) \in \set{I} \Rightarrow \exists N \in \Z_{\ge 0} :\\ \left(x(N) \in \set{F} \land \forall k \in [0;N[: x(k) \notin \set{O}\right) \},
\end{multline}
which enforces that all states in the initial set $\set{I}$ will reach the target $\set{F}$ in finite time while avoiding obstacles in~$\set{O}$. We use the abbreviated notation $\Sigma^{\text{Reach}} = [\set{I}, \set{F}, \set{O}]$ to denote the specification~\eqref{eq:reach_avoid}.
\hfill $\triangle$
\end{definition}

In addition to specifying the state behaviour of the closed-loop system, transition costs can be incorporated to address optimal control problems aiming to minimize the \emph{worst-case} cost.
\begin{definition}[Optimal control problem]\label{def:optimal_control_problem}
    Given a transition system $\Sys=(\set{X},\set{U},F)$, a reach-avoid specification $\Sigma^{\text{Reach}}=[\set{I},\set{F},\set{O}]$, and a transition cost function $c:\set{X} \times \set{U} \rightarrow \R_{>0}\cup \{\infty\}$, the \emph{optimal control problem} involves designing an \emph{optimal controller} $\Cont^*:\set{X}\rightarrow 2^{\set{U}}$ that minimizes the \emph{cost function}
    \begin{equation}\label{eq:cost_controller}
        \mathcal{L}(\Cont) =
        \begin{cases}
            \sup_{x_0 \in \set{I}} \ l_{\Cont}(x_0), & \text{if } \Cont \text{ solves } (\Sys,\Sigma^{\text{Reach}}),\\
            \infty, & \text{otherwise}.
        \end{cases}
    \end{equation}
    Here, $l_{\Cont}(x):\set{I}\rightarrow \R_{\ge 0}$ is defined as 
    \begin{equation}\label{eq:worst_case_cost_traj}
        l_{\Cont}(x_0) = \sup_{(\vect{x},\vect{u})\in \set{B}(\Cont\times \Sys)\mid x(0)=x_0} \sum_{i=0}^{N(\vect{x})-1} c(x(i),u(i)),
    \end{equation}
    where $N(\vect{x})=\min_{k\in \Z_{\ge 0}\mid x(k)\in \set{F}} k$, i.e., $x(N)$ is the first occurrence in the sequence $\vect{x}$ that belongs to $\set{F}$.
    \hfill $\triangle$
\end{definition}

Given a transition system $\Sys = (\set{X}, \set{U}, F)$, a \emph{value function} is a function $v : \set{X} \rightarrow \R_{\ge 0}\cup \{\infty\}$. Given a \emph{transition cost function} $c : \set{X} \times \set{U} \rightarrow \R_{>0} \cup \{\infty\}$ and a value function $v$, we denote by $B_c$ the \emph{Bellman operator} defined as
\begin{equation}\label{eq:bellman_operator}
    [B_c(v)](x) = \min_{u \in \set{U}_{\Sys}(x)} \left( c(x, u) + \max_{x^+ \in F(x, u)} v(x^+) \right).
\end{equation}
The \emph{controller associated} with a value function \(v\) is given by 
    $\Cont(x) = \argmin_{u \in \set{U}_{\Sys}(x)} \left( c(x,u) + \max_{x^+ \in F(x,u)} v(x^+) \right)$.

A value function $v$ that satisfies the Bellman equation~\cite{bellman1957dynamic}
\begin{equation}\label{eq:bellman_equation}
    v(x)= [B_c(v)](x), 
\end{equation}
is commonly referred to as a \emph{Bellman value function}. Its computation is generally intractable for infinite state systems. Instead, we consider \emph{suboptimal value functions}~\cite[Definition 7]{legat2021abstractionbased}.
\begin{definition}[Suboptimal value function]\label{def:back_lyapunov_like_value_function}
    Given a transition system $\Sys = (\set{X}, \set{U}, F)$ and a set $\set{X}_v \subseteq \set{X}$, a value function $v : \set{X} \rightarrow \R_{\ge 0}\cup \{\infty\}$ is a \emph{suboptimal value function} with transition cost function $c : \set{X} \times \set{U} \rightarrow \R_{> 0}\cup \{\infty\}$, if for all $x \in \set{X} \setminus \set{X}_v$, $v(x) = \infty$, and for all $x \in \set{X}_v$, $v(x)$ is finite and $v(x) \ge [B_c(v)](x)$.
    \hfill $\triangle$
\end{definition}

The following proposition demonstrates that the controller associated with a suboptimal value function provides an upper bound on the optimal worst-case cost (see Definition~\ref{def:optimal_control_problem}). This proposition is inspired by works such as \cite{meyn2022CS&RL}. However, the variant we propose for robust control of set-valued dynamics is, to our knowledge, new to the literature.
\begin{proposition}\label{th:controller_from_value_function}
    Consider a transition system \(\Sys = (\set{X}, \set{U}, F)\), a reach-avoid specification \(\Sigma^{\text{Reach}} = [\set{I}, \set{F}, \set{O}]\), and a transition cost function \(c : \set{X} \times \set{U} \rightarrow \R_{>0} \cup \{\infty\}\), such that there exists \(b > 0\) with \(c(x, u) \ge b\) for all \(x \in \set{X}\) and all \(u \in \set{U}_{\Sys}(x)\).
    Let \(v : \set{X} \rightarrow \R_{\ge 0} \cup \{\infty\}\) be a suboptimal value function such that $v(x)= 0$ if and only if $x \in \set{F}$, $v(x)= \infty$ if $x \in \set{O}$, and $v(x) > 0$ otherwise. 
    
    Then, the controller \(\Cont\) associated with \(v\) solves the control problem \((\Sys, [\dom(\Cont), \set{F}, \set{O}])\). If \(\set{I} \subseteq \dom(\Cont)\), then \(\Cont\) also solves \((\Sys, \Sigma^{\text{Reach}})\). Moreover, the value function \(v(x)\) provides an upper bound on the worst-case cost under \(\Cont\), i.e., \(v(x) \ge \ell_{\Cont}(x)\), and the overall performance satisfies \(\mathcal{L}(\Cont) \ge \mathcal{L}(\Cont^*)\), with equality if \(v\) is a Bellman value function.\hfill $\triangle$
\end{proposition}
\begin{proof}
    Since $v$ is a suboptimal value function, for all $x \in \set{X}_v$, where $\set{X}_v = \{x \in \set{X} \mid v(x) < \infty\}$, we have $v(x) \ge [B_c(v)](x)$. For $u \in \Cont(x)$, it holds that $\max_{x^+ \in F(x,u)} v(x^+) \le v(x) - c(x,u) \le v(x)-b < v(x)$ since $c(x,u) \ge b$. As a result, the function $v$ decreases by at least $b$ along trajectories of $\Cont \times \Sys$ for all $x \in \set{X}_v$ towards the minimum of $v$, which corresponds to the target set $\set{F}$ while avoiding the obstacle $\set{O}$.
\end{proof}

The following theorem formalizes the transfer of surrogate value functions from the abstract system to the concrete system.
\begin{theorem}[{\cite[Theorem 1]{egidio2022state}}]\label{th:construct_lyapunov_like_value_function}
    Consider two transition systems $\Sys_1=(\set{X}_1,\set{U}_1,F_1)$ and $\Sys_2=(\set{X}_2,\set{U}_2,F_2)$, and a relation $R \subseteq \set{X}_1 \times \set{X}_2$ such that $\Sys_1 \preceq_R^{\FRR} \Sys_2$. Given a transition cost $c_1$ for $\Sys_1$, we consider a transition cost $c_2$ for $\Sys_2$ such that
    \begin{equation}\label{eq:abstract_cost_ASR_det_Lyap}
        c_2(x_2, u) \ge \max_{x_1\in R^{-1}(x_2)} c_1(x_1, u).
    \end{equation}
    If $v_2$ is a suboptimal value function for $\Sys_2$ with transition cost $c_2$, then
    \begin{equation}\label{eq:lyap_from_abstract}
        v_1(x_1) = \min_{x_2 \in R(x_1)} v_2(x_2)
    \end{equation}
    is a suboptimal value function with transition cost $c_1$ for~$\Sys_1$.
\end{theorem}

In the special case where \(R\) is deterministic, the value function defined in~\eqref{eq:lyap_from_abstract} reduces to \(v_1 = v_2 \circ R\), as previously established in~\cite[Theorem 8]{legat2021abstractionbased}. This is the case in the classical abstraction approach \cite{reissig2016feedback, tabuada2009verification}, where the abstract state space corresponds to a partition of (a subset of) the concrete state space. In practice, while the finiteness of \(\Sys_2\) enables efficient computation of a Bellman value function \(v_2\) using classical graph-theoretic tools such as dynamic programming~\cite[Chapter 1]{bertsekasvol1}, Theorem~\ref{th:construct_lyapunov_like_value_function} allows it to be translated into a suboptimal value function \(v_1\) for the concrete system.

\section{NUMERICAL EXPERIMENTS}
\label{sec:experiments}

In this section, we illustrate the benefits of our approach on two different optimal control problems. 

Before stating the numerical results, we outline the abstraction construction. We construct our abstractions $\Sys_2$ in the same fashion as in \cite{reissig2016feedback}, but allow multiple timesteps. To construct the abstract system \(\Sys_2 = (\set{X}_2, \set{U}_2, F_2)\), we consider a uniform partition of the concrete state space \(\set{X}_1\). Formally, the abstract state space is defined as $\set{X}_2 := \{q_1, \dots, q_{N_q}\}$, where \(N_q\) is the number of partition cells. The relation \(R \subseteq \set{X}_1 \times \set{X}_2\) is defined such that \((x_1, q) \in R\) if and only if \(x_1 \in \set{X}_1\) lies in the cell represented by \(q \in \set{X}_2\). The transition function \(F_2\) is computed via a growth bound~\cite[Definition VIII.2]{reissig2016feedback} to over-approximate the concrete \emph{reachable sets} $F_1(R^{-1}(q), (u, \tau))$ for all $q \in \set{X}_2$ and $(u, \tau) \in \set{U}_2$. By Theorem~\ref{th:formal_abstraction}, this ensures that \(\Sys_1 \preceq_R^{\FRR} \Sys_2\). The abstract reach-avoid specification \(\Sigma_2^{\text{Reach}}\) is constructed analogously to~\cite{reissig2016feedback}, ensuring that \(\Sigma_1^{\text{Reach}} \subseteq R^{-1}(\Sigma_2^{\text{Reach}})\). Then, by Proposition~\ref{prop:FRR}, any controller \(\Cont_2\) solving \((\Sys_2, \Sigma_2^{\text{Reach}})\) can be concretized into a controller \(\Cont_1 = \Cont_2 \circ R\) that solves \((\Sys_1, \Sigma_1^{\text{Reach}})\).
Finally, to ensure a fair comparison, we keep abstractions of the same size by fixing the number of input-timesteps pairs.

\subsection{Effort-optimal control of a simple pendulum}

To illustrate the theoretical framework of the previous section, we consider a frictionless simple pendulum. This continuous time system is given by \(\set{D}=(\set{X}, \set{U}, f)\) , with \(\set{X} = [-\pi, 2\pi]\), \(\set{U} = [-6,6] \) and dynamics defined by
$f((\theta,\omega),u) = \begin{pmatrix}
    \omega \\
    -\frac{g}{l}\sin{(\theta) + u}
\end{pmatrix},$
where $g=9.81[m/s^2]$ is gravity, $l$ is the length of the rod and $u$ is the input torque. 

Given the transition system $\set{S}_1 = (\set{X}_1, \set{U}_1, F_1)$ associated with $\set{D}$ as in Definition~\ref{def:multisam}, we define the concrete specification $\Sigma_1^{\text{Reach}}=[\set{I}_1, \set{F}_1, \set{O}_1]$ with $\set{I}_1 = [-\frac{5\pi}{180}, \frac{5\pi}{180}] \times[-0.5,0.5]$, $\set{F}_1=[\pi-\frac{17.5}{180}, \pi+\frac{17.5\pi}{180}]\times[-1,1]$ and $\set{O}_1=\myemptyset$. 

We consider a cost function $c_1(x_1,(u, \tau)) := u^2\tau + \epsilon$\footnote{The term $\epsilon$ is added to meet the assumptions of Proposition~\ref{th:controller_from_value_function}, as without this, $u=0$ leads to a zero cost. In practice, this $\epsilon$ can be chosen arbitrarily close to zero.}. One can interpret this cost as a measure of the effort that is needed to reach the specification. The abstract transition cost is then given by $c_2(q, (u,\tau)) := u^2\tau+ \epsilon$, which satisfies condition~\eqref{eq:abstract_cost_ASR_det_Lyap}, since \(c_1\) is independent of the concrete state~\(x_1\). As a result, by Theorem~\ref{th:construct_lyapunov_like_value_function}, given a Bellman value function \(v_2\) and its associated controller \(\Cont_2\) for \(\Sys_2\), one can construct a suboptimal value function \(v_1 = v_2 \circ R\) for \(\Sys_1\), which provides an upper bound on the worst-case cost under the concretized controller \(\Cont_1 = \Cont_2 \circ R\).

Results are summarized in Table~\ref{tab:comparison_effort}. We can observe that with our method we get a feasible abstract problem, which outperforms controllers synthesized with a single timestep. An example trajectory as well as a representation of the suboptimal value function is given in Figure~\ref{fig:optimaleffort}.
\begin{figure}
    \centering    \includegraphics[width=0.75\linewidth]{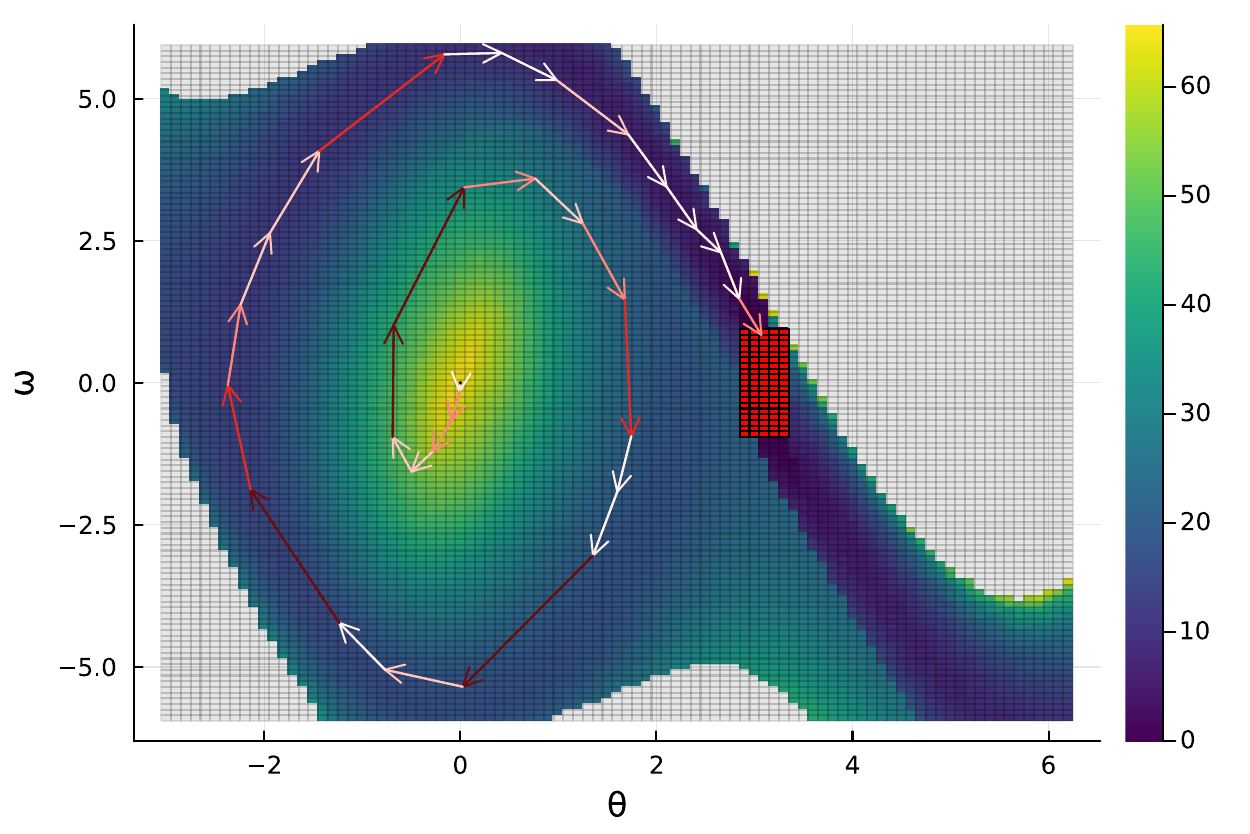}
    \caption{Controlled trajectory of the simple pendulum. The red region denotes the target set \(\set{F}_1\). The heatmap shows the computed suboptimal value function \(v_1\) for \(\Sys_1\). Arrow colors indicate the timestep length used along the closed-loop trajectory—the lighter, the shorter.}
    \label{fig:optimaleffort}
\end{figure}

\begin{table}
\centering
\small
\resizebox{0.80\linewidth}{!}{%
 \begin{tabular}{c|c}
        \toprule
         Timestep(s) [s] & Worst-case effort guarantee ($\mathcal{L}(\Cont_1)$) \\
        \midrule 
        0.1 & Infeasible \\ 
        0.15 & 81.47 \\
        0.2 & Infeasible \\
        0.25 & Infeasible \\
        0.3 & Infeasible \\
        0.1, 0.15, 0.2, 25, 0.3 & 65.70 \\
        \bottomrule
    \end{tabular}
}
\caption{Worst-case effort guarantees for the simple pendulum problem for different individual sampling steps. The last row shows the results of our method combining multiple timesteps, with a coarser input space discretization.}
\label{tab:comparison_effort}
\end{table}

\subsection{Time-optimal path-planning}
We consider a path-planning problem in which a vehicle must navigate through a two-dimensional maze. The system has three state variables: the position \((x, y)\) in the plane and the orientation \(\theta\). It is controlled via two inputs: wheel speed \(u_1\) and steering angle \(u_2\). This model has previously been studied in the context of abstraction-based control~\cite{reissig2016feedback}.

The continuous-time system is given by \(\mathcal{D} = (\set{X}, \set{U}, f)\), with state space \(\set{X} = [-0.1,10.1] \times [-0.1,10.1] \times [-\pi-0.4, \pi+0.4]\), input space \(\set{U} \subseteq [-1,1] \times [-1,1] \), and dynamics defined by
\(
f((x, y, \theta), (u_1, u_2)) =
\begin{pmatrix}
    u_1 \cos(\alpha + \theta) \cos(\alpha)^{-1} \\
    u_1 \sin(\alpha + \theta) \cos(\alpha)^{-1} \\
    u_1 \tan(u_2)
\end{pmatrix},
\)
where \(\alpha = \arctan\left( \frac{\tan(u_2)}{2} \right)\).

Given the transition system \(\Sys_1 = (\set{X}_1, \set{U}_1, F_1)\) associated with \(\mathcal{D}\), we define the concrete specification \(\Sigma_1^{\text{Reach}}\), which constrains only the first two state variables, as illustrated in Figure~\ref{fig:path_planning} and its caption.


We consider the concrete transition cost \(c_1(x_1, (u, \tau)) := \tau\), which corresponds to minimizing the total time required to satisfy the reach-avoid specification. The abstract transition cost is \(c_2(q, (u, \tau)) := \tau\), which satisfies condition~\eqref{eq:abstract_cost_ASR_det_Lyap}, since \(c_1\) is independent of the concrete state \(x_1\).

\begin{table}
\centering
\small
\resizebox{0.85\linewidth}{!}{
 \begin{tabular}{c|c}
        \toprule
         Timestep(s) [s] & Worst-case time guarantee ($\mathcal{L}(\Cont_1)$) [s] \\
        \midrule
        0.2 & Infeasible \\
        0.5 & Infeasible \\
        0.8 & 144.0 \\
        0.2, 0.5, 0.8 & 114.60 \\
        \midrule
        0.3, 0.7, 1.1, 1.5, 1.9 & 94.1 \\
        \bottomrule
    \end{tabular}
}
\caption{Worst-case time guarantees for the path planning problem, for different individual sampling steps. The last 2 rows shows the result of our method combining multiple timesteps, with a coarser input space discretization. The last row is obtained with the data-driven approach (Subsection~\ref{subsec:data}.}
\label{tab:comparison}
\end{table}
Results are summarized in Table~\ref{tab:comparison}. Again, a combination of timesteps is used and the worst-case cost is improved. See on Figure~\ref{fig:path_planning} an example trajectory obtained with our approach.
\begin{figure}
    \centering
    \includegraphics[width=0.75\columnwidth]{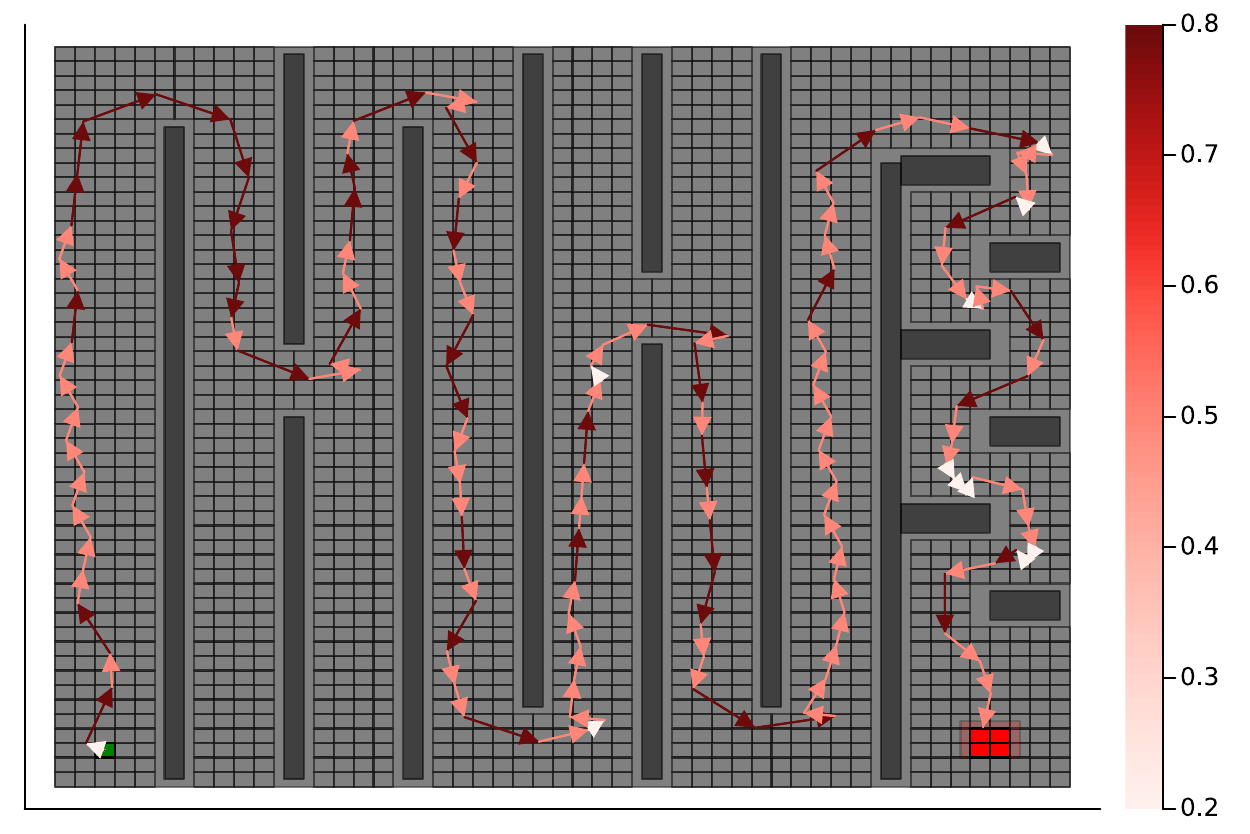}
    \caption{An example trajectory, projected onto the first two components of the system, for the path-planning problem.
    The initial set $\set{I}_1 = \{x_0\}$, target set $\set{F}_1$, and obstacle set $\set{O}_1$ are shown in green, red, and black, respectively.
    The color of the arrow indicates the used timestep. 
    }
    \label{fig:path_planning}
\end{figure}

\subsection{Data-driven approach} \label{subsec:data}

Growth-bound functions provide over-approximations of the reachable sets $F_1(R^{-1}(q), (u, \tau))$, necessary so that Theorem~\ref{th:formal_abstraction} holds. However, they become more and more conservative as the timestep increases \cite{reissig2016feedback}. 
This can lead to infeasibility of the abstract problem by introducing non-determinism, and prevents us from experimenting with longer timesteps, as the over-approximation may intersect obstacles and be discarded.

To mitigate this, we propose a data-driven abstraction approach. For each abstract state-input pair $(q,(u,\tau))$ of the abstract system, we sample $N$ points $x_1 \in R^{-1}(q)$. We observe $x_1^+\in F_1(x_1,(u,\tau))$, and if there exists $q^+ \in \set{X}_2$ such that $(x_1^+,q^+)\in R$, we add a transition $q^+\in F_2(q,(u, \tau))$.

This has the effect of \emph{under-approximating} the reachable set, as we may overlook a trajectory. 
As a result, Theorem~\ref{th:formal_abstraction} no longer applies. Using this data-driven approach, we are able to consider larger timesteps, as shown in Table~\ref{tab:comparison}, which leads to improved expected worst-case cost.  
In Figure~\ref{fig:path_planning_data}, we show a sample trajectory obtained with this method. As with the growth-bound approach, the controller uses a combination of different timesteps.


\begin{figure}
    \centering
\includegraphics[width=0.75\linewidth]{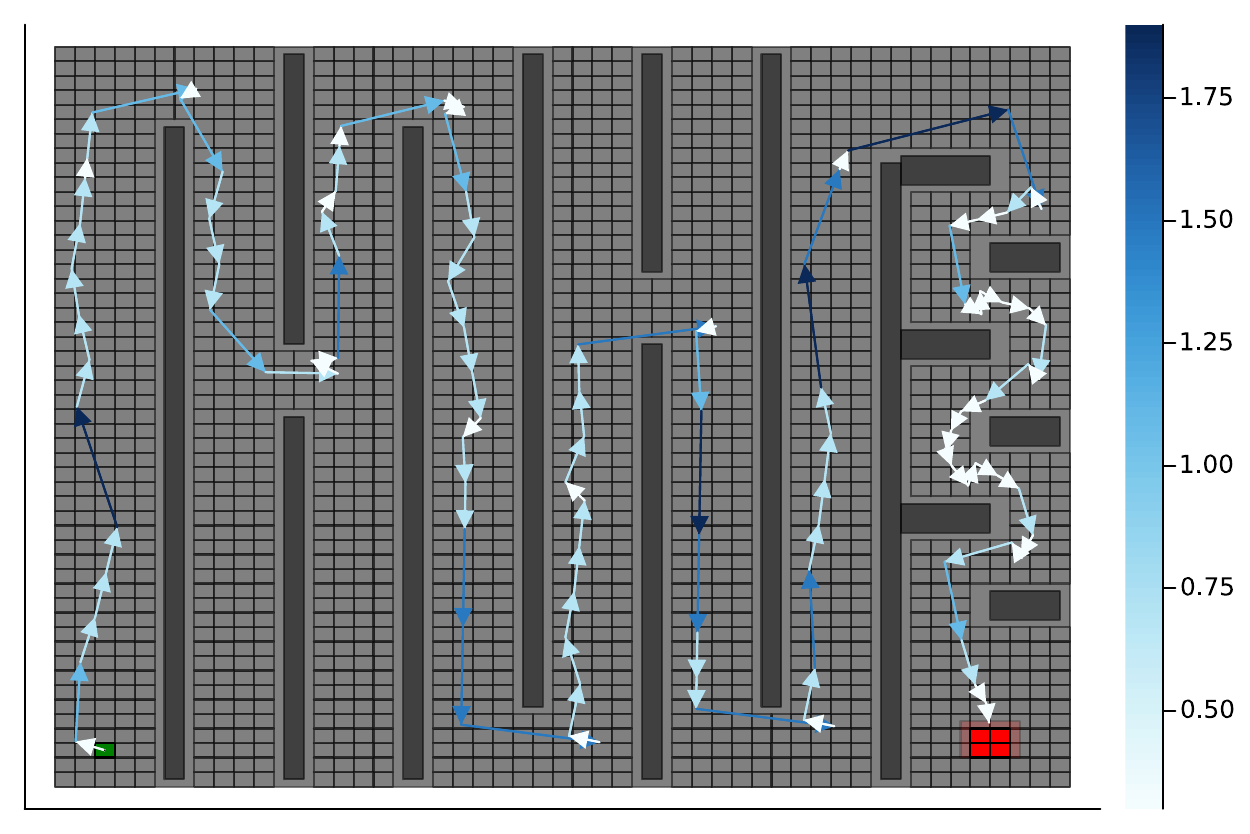}
    \caption{An example trajectory for the path-planning problem for the \textbf{data driven approach}. The color of the arrows indicate the duration of the timestep. 
    }
    \label{fig:path_planning_data}
\end{figure}



\section{CONCLUSIONS AND FUTURE WORKS}

In this work, we introduced a temporal abstraction framework for abstraction-based control of continuous-time systems. By allowing multiple timesteps within a unified abstraction, we addressed the trade-off between conservatism and feasibility that arises in fixed-step discretizations. We proved that these abstractions preserve the feedback refinement relation, enabling formal guarantees to be carried over from the abstract to the concrete system. We further demonstrated the effectiveness of our approach through numerical experiments, showing improved worst-case performance without increasing the size of the abstractions.



Future work includes optimizing time and space discretization, reducing conservatism in reachability computation, and exploring probabilistic guarantees in data-driven settings.


\begin{spacing}{1.0}
    \bibliographystyle{abbrv}
    \bibliography{ref}
\end{spacing}

\end{document}

%% file: header.tex
\usepackage{float}
\usepackage{amsmath}
\usepackage{amssymb}

\usepackage{algorithm2e}

\usepackage{amsfonts}
\usepackage{ifthen}
\usepackage{xcolor}

\usepackage{caption}

\usepackage{hyperref}
\usepackage{cleveref}

\usepackage{comment}

%% file: macros.tex
\newcommand{\jc}[1]{\textcolor{blue}{[JC]: #1}}

\newcommand{\modif}[1]{\ifthenelse{\equal{\discussion}{true}}{\textcolor{red}{#1}}}

\newcommand{\vect}[1]{\boldsymbol{#1}} 

\newcommand{\R}{\mathbb{R}}
\newcommand{\N}{\mathbb{N}}

\newcommand{\Z}{\mathbb{Z}}

\newcommand{\myemptyset}{\varnothing}
\newcommand{\argmin}{\operatorname{argmin}} 
\newcommand{\dom}{\operatorname{dom}} 

\newtheorem{theorem}{Theorem}
\newtheorem{proposition}{Proposition}
\newtheorem{definition}{Definition}

\newcommand{\set}[1]{\mathcal{#1}}
\newcommand{\Sys}{\mathcal{S}}
\newcommand{\Cont}{\mathcal{C}}

\newcommand{\FRR}{\operatorname{FRR}}

%% file: root.bbl
\begin{thebibliography}{10}

\bibitem{althoff2014online}
M.~Althoff and J.~M. Dolan.
\newblock Online verification of automated road vehicles using reachability analysis.
\newblock {\em IEEE Transactions on Robotics}, 30(4):903--918, 2014.

\bibitem{althoff2021set}
M.~Althoff, G.~Frehse, and A.~Girard.
\newblock Set propagation techniques for reachability analysis.
\newblock {\em Annual Review of Control, Robotics, and Autonomous Systems}, 4(1):369--395, 2021.

\bibitem{alur2015principles}
R.~Alur.
\newblock {\em Principles of cyber-physical systems}.
\newblock MIT press, 2015.

\bibitem{bellman1957dynamic}
R.~Bellman.
\newblock {\em Dynamic Programming}.
\newblock Dover Publications, 1957.

\bibitem{bertsekasvol1}
D.~Bertsekas.
\newblock {\em Dynamic programming and optimal control: Volume I}, volume~1.
\newblock Athena scientific, 2005.

\bibitem{calbert2024classification}
J.~Calbert, A.~Girard, and R.~M. Jungers.
\newblock Classification of simulation relations for symbolic control.
\newblock {\em arXiv preprint arXiv:2410.06083}, 2024.

\bibitem{egidio2022state}
L.~N. Egidio, T.~A. Lima, and R.~M. Jungers.
\newblock State-feedback abstractions for optimal control of piecewise-affine systems.
\newblock In {\em 2022 IEEE 61st Conference on Decision and Control (CDC)}, pages 7455--7460. IEEE, 2022.

\bibitem{gillula2011applications}
J.~H. Gillula, G.~M. Hoffmann, H.~Huang, M.~P. Vitus, and C.~J. Tomlin.
\newblock Applications of hybrid reachability analysis to robotic aerial vehicles.
\newblock {\em The International Journal of Robotics Research}, 2011.

\bibitem{hartman2002ordinary}
P.~Hartman.
\newblock {\em Ordinary differential equations}.
\newblock SIAM, 2002.

\bibitem{ivanova2020lazy}
E.~Ivanova and A.~Girard.
\newblock Lazy safety controller synthesis with multi-scale adaptive-sampling abstractions of nonlinear systems.
\newblock {\em IFAC-PapersOnLine}, 53(2):1837--1843, 2020.

\bibitem{kader2019safety}
Z.~Kader, A.~Saoud, and A.~Girard.
\newblock Safety controller design for incrementally stable switched systems using event-based symbolic models.
\newblock In {\em 18th European Control Conference (ECC)}, pages 1269--1274. IEEE, 2019.

\bibitem{kim2012cyber}
K.-D. Kim and P.~R. Kumar.
\newblock Cyber--physical systems: A perspective at the centennial.
\newblock {\em Proceedings of the IEEE}, 100(Special Centennial Issue):1287--1308, 2012.

\bibitem{knight2002safety}
J.~C. Knight.
\newblock Safety critical systems: challenges and directions.
\newblock In {\em Proceedings of the 24th international conference on software engineering}, pages 547--550, 2002.

\bibitem{lee2009introducing}
E.~A. Lee.
\newblock Introducing embedded systems: a cyber-physical approach.
\newblock In {\em Proceedings of the 2009 Workshop on Embedded Systems Education}, pages 1--2, 2009.

\bibitem{legat2021abstractionbased}
B.~Legat, J.~Bouchat, and R.~M. Jungers.
\newblock Abstraction-based branch and bound approach to q-learning for hybrid optimal control.
\newblock In {\em Proceedings of the 3rd Annual Learning for Dynamics \& Control Conference}, 2021.

\bibitem{meyn2022CS&RL}
S.~Meyn.
\newblock {\em Control Systems and Reinforcement Learning}.
\newblock 2022.

\bibitem{reissig2016feedback}
G.~Reissig, A.~Weber, and M.~Rungger.
\newblock Feedback refinement relations for the synthesis of symbolic controllers.
\newblock {\em IEEE Transactions on Automatic Control}, 62(4):1781--1796, 2016.

\bibitem{sutton1999between}
R.~S. Sutton, D.~Precup, and S.~Singh.
\newblock Between mdps and semi-mdps: A framework for temporal abstraction in reinforcement learning.
\newblock {\em Artificial intelligence}, 112(1-2):181--211, 1999.

\bibitem{tabuada2009verification}
P.~Tabuada.
\newblock {\em Verification and control of hybrid systems: a symbolic approach}.
\newblock Springer Science \& Business Media, 2009.

\end{thebibliography}
